%% file: zafari_coalition_journal.tex
\documentclass[10pt,journal]{IEEEtran}
\newcommand{\subparagraph}{}
\usepackage[utf8]{inputenc}
\usepackage{graphicx}
\usepackage{amsmath}
\usepackage{amssymb}
\usepackage{amsthm}
\usepackage{algpseudocode}
\usepackage{epstopdf}
\usepackage{algorithm}
\newcommand\NoDo{\renewcommand\algorithmicdo{}}

\usepackage{color}
\usepackage{comment}
\usepackage{multirow}
\usepackage{url}
\usepackage{tabu}
\usepackage{textcomp}
\newtheorem{definition}{Definition}
\newtheorem{theorem}{Theorem}

\newtheorem{remark}{Remark}

\newtheorem{lemma}{Lemma}

\usepackage{ifpdf}
\ifpdf
\usepackage{epstopdf}
\fi
\usepackage{titlesec}
\usepackage{amssymb}

\renewcommand\algorithmicdo{}

\let\emptyset\varnothing
\ifCLASSOPTIONcompsoc
  \usepackage[nocompress]{cite}
\else
  \usepackage{cite}
\fi

\ifCLASSINFOpdf
\else
\fi

\begin{document}

\title{Let's Share: A Game-Theoretic Framework for Resource Sharing  in Mobile Edge Clouds}

\author{Faheem~Zafari,~\IEEEmembership{Student Member,~IEEE,}
       Kin~K. Leung,~\IEEEmembership{Fellow,~IEEE,} 
        Don~Towsley,~\IEEEmembership{Life Fellow,~IEEE,}
          \\
        Prithwish~Basu,~\IEEEmembership{Senior Member,~IEEE,}
         Ananthram~Swami,~\IEEEmembership{Fellow,~IEEE, }
          and~Jian~Li,~\IEEEmembership{Member,~IEEE}\thanks{This paper has been presented in part at the ACM Mobicom Workshop on Technologies for the Wireless Edge, October 2018, India. }
\IEEEcompsocitemizethanks{\IEEEcompsocthanksitem Faheem Zafari and Kin K. Leung are  with the Department
of Electrical and Electronics Engineering, Imperial College London, London,
UK.   \protect\\
E-mail:  \{faheem16,kin.leung\}@imperial.ac.uk
\IEEEcompsocthanksitem  Don Towsley is with College of Information and Computer Sciences, University of Massachusetts Amherst, Amherst, MA 01003, USA.\protect\\
E-mail:  towsley@cs.umass.edu
\IEEEcompsocthanksitem Prithwish Basu is with BBN Technologies,  MA, USA.
\protect\\
Email: prithwish.basu@raytheon.com
\IEEEcompsocthanksitem Ananthram Swami is with the U.S. Army Research Laboratory, Adelphi, MD 20783, USA.
\protect\\
Email: ananthram.swami.civ@mail.mil
\IEEEcompsocthanksitem Jian Li is with  is with the Department of Electrical and Computer Engineering, Binghamton University, the State University of New York Binghamton, NY 13902, USA.
\protect\\
Email: lij@binghamton.edu

}}

\markboth{}
{Zafari \MakeLowercase{\textit{et al.}}: Bare Demo of IEEEtran.cls for Computer Society Journals}

\IEEEtitleabstractindextext{%
\begin{abstract}
Mobile edge computing seeks to provide resources to different delay-sensitive applications. This is a challenging problem as an edge cloud-service provider may not have sufficient resources to satisfy all resource requests. Furthermore, allocating available resources optimally to different applications is also challenging. Resource sharing among different edge cloud-service providers can address the aforementioned limitation as certain service providers may have resources available that can be ``rented'' by other service providers.  However, edge cloud service providers can have different  objectives or \emph{utilities}. Therefore, there is a need for  an efficient and effective mechanism to share resources among service providers, while considering the different objectives of various providers. 
We model resource  sharing as a multi-objective optimization problem and
present a solution framework based on  \emph{Cooperative Game Theory} (CGT). We consider the strategy where each service provider allocates  resources to its native applications first and shares the remaining resources with applications from other service providers. We prove that for a monotonic, non-decreasing utility function, the game is canonical and convex. Hence, the \emph{core} is not empty and the grand coalition is stable. We propose two  algorithms \emph{Game-theoretic Pareto optimal allocation} (GPOA) and \emph{Polyandrous-Polygamous Matching based Pareto Optimal Allocation} (PPMPOA) that  provide allocations from the core. 
Hence the  obtained allocations are \emph{Pareto} optimal and the grand coalition of all the service providers is stable.  
Experimental results confirm that our proposed resource sharing framework improves utilities of edge cloud-service providers and application request  satisfaction. 
\end{abstract}

}

\maketitle

\IEEEdisplaynontitleabstractindextext
\IEEEpeerreviewmaketitle

\section{Introduction}\label{sec:intro}
\input{01-intro2}

\section{Preliminaries}\label{sec:prelim}
\input{preli}
\section{System Model}\label{sec:sysmodel}
\input{02-sysmodel2}

\section{Game Theoretic Solution}\label{sec:opt_problem}

\input{03-opt-problem}
	\vspace{-0.3in}
\section{Experimental Results}\label{sec:exp_results}
\input{04-exp-results}

	\vspace{-0.10in}
\section{Related Work}\label{sec:related}

\input{related}

	\vspace{-0.10in}
\section{Conclusions}\label{sec:conclusion}
\input{05-conclusions}

  \section*{Acknowledgments}
  
	This work was supported by the U.S. Army Research Laboratory and the U.K. Ministry of Defence under Agreement Number W911NF-16-3-0001. The views and conclusions contained in this document are those of the authors and should not be interpreted as representing the official policies, either expressed or implied, of the U.S. Army Research Laboratory, the U.S. Government, the U.K. Ministry of Defence or the U.K. Government. The U.S. and U.K. Governments are authorized to reproduce and distribute reprints for Government purposes notwithstanding any copy-right notation hereon.  Faheem Zafari also acknowledges the financial support by EPSRC Centre for Doctoral Training in High Performance Embedded and Distributed Systems (HiPEDS, Grant Reference EP/L016796/1), and Department of Electrical and Electronics Engineering, Imperial College London.

\ifCLASSOPTIONcaptionsoff
  \newpage
\fi

\bibliographystyle{unsrt}
\bibliography{refs}

\end{document}

%% file: 01-intro2.tex
Mobile edge computing is a viable solution to support resource intensive applications (users)\footnote{Throughout the paper, we use the terms ``applications'' and ``users'' interchangeably.} that relies on mobile edge clouds (MECs) placed at the edge of any network \cite{he2018s}. This, in contrast with running applications on user devices (with limited processing capability) or deep in the Internet, usually allows  one-hop communication \cite{jia2015optimal} between  edge clouds and applications in order to  provide the required processing power and reduce application latency. 
The goal for any MEC service provider  is to 
provide resources to applications as much as possible, resulting in an increase in application  satisfaction and \emph{utility} for the provider.
However, a fundamental limitation of MEC is that in contrast to traditional cloud platforms and data centers,  edge clouds have  limited resources and may not always be able to satisfy  application demands for resources \cite{he2018s}.  In addition to commercial settings, this problem of limited resources also occurs in military networks.  For example, MEC service providers are equivalent to military coalition partners, who own various parts of the computing resources.
\par  To address  resource scarcity, different solutions have been proposed in the literature such as purchasing more resources to meet  peak demands, using resources from the same service provider located at a distant location, or  sharing resources among  co-located mobile edge clouds. However, purchasing more resources to meet peak demands is not an efficient solution as the newly purchased resources may not always be used. Furthermore, it is not always possible to accurately predict peak demand. Hence, delays  in purchasing additional resources may introduce delays in satisfying the applications.  Similarly, using resources of the same provider at a distant location can increase latency. Creating shared resources among  co-located (or physically adjacent)   edge clouds is a promising approach as it accounts for the aforementioned shortcomings of buying more resources to meet peak demands or allocating resources at a distant location. However,  resources of  co-located mobile edge clouds belong to different service providers, each of which  may have a particular objective to optimize such as security, throughput, latency, etc.  when allocating and utilizing its resources. Therefore creating a resource pool and then allocating these resources requires taking different service provider objectives into account, which results in a multi-objective optimization \cite{cho2017survey} (MOO) problem.

\par In this paper, we attempt to address these issues. Specifically,  we consider a setting where multiple edge clouds reside at the network edge  belonging to different service providers. Each
MEC has a specific set and amount of resources and a number of applications
affiliated with the MEC demand for resources. We consider
 a strategy where all edge clouds initially allocate resources
 to their native applications. If an edge cloud satisfies its own
 applications and has free resources, it can share them with
 other edge clouds that might need them. This is because
 satisfying native applications helps the MEC service provider
 to retain its customers. We present a resource sharing and allocation framework based on \emph{cooperative game theory} 
(CGT), in
 which different edge clouds form
 a coalition to share their resources and satisfy the requests of different applications.
 Our CGT-based framework considers the fact that different
edge clouds may have different objectives, which makes the
 traditional single-objective optimization framework inapplicable to the MEC settings.



\subsection{Summary of the contributions}
The main contributions of this paper are:
\begin{enumerate}
	\item 
	We present a novel multi-objective resource-sharing
	problem formulation that considers both
	service provider utilities and application request satisfaction
	in mobile edge clouds. 
	 We propose a cooperative game theory based framework to solve the formulated problem. We model each edge cloud service provider as a player in our game and  model  resource sharing and allocation problem  as a \emph{canonical} game with non-transferable utility (NTU). 
	 We prove that the game is super-additive and convex, hence the \emph{core} is non-empty. Therefore, the grand coalition of all edge clouds is stable.  
\item We   propose  a distributed \emph{Game-theoretic Pareto Optimal Allocation} (GPOA) algorithm that provides an allocation from the core based on a strategy where all MECs first allocate resources to their native applications, and then MECs with resource surplus share their resources with those that have a resource deficit. 
\item To reduce resource fragmentation\footnote{We define resource fragmentation as splitting of the provisioned resources to applications of any player $n$ across different MECs. The higher the number of MECs that provide resources to a particular application, the higher will be the resource fragmentation. Ideally, for the applications of a resource deficit MEC, only one MEC with resource surplus should provide the resources. } due to resource sharing among MECs, we also propose  \emph{Polyandrous Polygamous Matching based Pareto Optimal Allocation} (PPMPOA) algorithm that aims to match resource deficit MECs with resource surplus MECs. 
Using evaluations we show that PPMPOA:
\begin{itemize}
	\item Provides allocation from the core.
	\item Results in a \emph{stable} matching \cite{maschler2013}.
	\item Reduces resource fragmentation. 
\end{itemize}
\item We also prove that the proposed  algorithms, GPOA and PPMPOA,  enforce \emph{truth-telling}, i.e., no MEC service provider has the incentive to misreport its capacity and native application requests. 

	
	\item We evaluate the performance of our proposed framework
	by extensive simulations. We verify the game
	theoretic properties of our proposed approach by showing
	that the game is super-additive and the core is nonempty,
	i.e., our obtained solutions are from the core
	 and satisfy properties such as \emph{individual rationality}, \emph{group rationality}, and \emph{stability}.  We also  show that the resource sharing and allocation mechanism  improves the utilities of  game players. Furthermore, our framework also improves application (user) satisfaction.
	 
\end{enumerate}
The rest of the paper is  structured as follows. 
Section~\ref{sec:prelim} provides a primer on  MOO,  CGT, and  the core. 
 Section~\ref{sec:sysmodel} describes our system model and presents the resource sharing/allocation optimization problem.  Section \ref{sec:opt_problem} discusses the game theoretic solution  
and Section~\ref{sec:exp_results} presents our experimental results.  Section \ref{sec:related} presents a review of relevant literature, while Section~\ref{sec:conclusion} concludes the paper.  


%% file: preli.tex
In this section, we give a brief overview of Multi-Objective Optimization (MOO), cooperative game theory, and the core.

\subsection{Multi-Objective Optimization}
 Multi-Objective Optimization (MOO) identifies a vector $\boldsymbol{x}^*=[x_1^*,x_2^*,\cdots ,x_t^*]^T$ that optimizes a vector function
	\begin{equation}\label{eq:vectorobjective}
	\centering 
	\bar{f}(\boldsymbol{x})=[f_1(\boldsymbol{x}),f_2(\boldsymbol{x}),\cdots,f_N(\boldsymbol{x})]^T
	\end{equation}
	such that 
	\begin{align}\label{eq:constraints}
	\centering 
	& g_i(\boldsymbol{x})\geq 0, \; i=1,2,\cdots,m, \\
	&h_i(\boldsymbol{x})=0\;\; i=1,2,\cdots,p. \nonumber
	\end{align}
	where $\boldsymbol{x}=[x_1,x_2,\cdots ,x_t]^T$ is a vector of $t$ decision variables and the feasible set is denoted by $F$.
The fundamental difference between a single objective optimization (SOO) and MOO is that MOO involves a vector of objective functions rather than a single objective function. Hence, the resulting solution is not a single point, but 
 a \emph{frontier} of solutions known as \emph{Pareto frontier} or \emph{Pareto boundary} (see \cite{cho2017survey} for details). Some basic definitions related to MOO are: 
\begin{definition}{Pareto Improvement or Pareto Dominated Solution:}
	Given an initial allocation, if we can achieve a different 	allocation that improves at least one individual function  without degrading any other, then the initial allocation is called \emph{Pareto improvement}.
\end{definition}

\begin{definition}{Pareto Optimality:}
	For any minimization problem,  $\boldsymbol{x}^*$ is $Pareto\; Optimal$ if the following holds for every $\boldsymbol{x} \in F$,
	\begin{align}\label{eq:paretoptimality}
	\bar{f}(\boldsymbol{x}^*)\leq\bar{f}(\boldsymbol{x}). 
	\end{align}
	where 
	$\bar{f}(\boldsymbol{x}^*)=[f_1(\boldsymbol{x}^*),f_2(\boldsymbol{x}^*),\cdots,f_N(\boldsymbol{x}^*)]^T$.
\end{definition}

\begin{figure}
\centering
\includegraphics[width=0.4\textwidth]{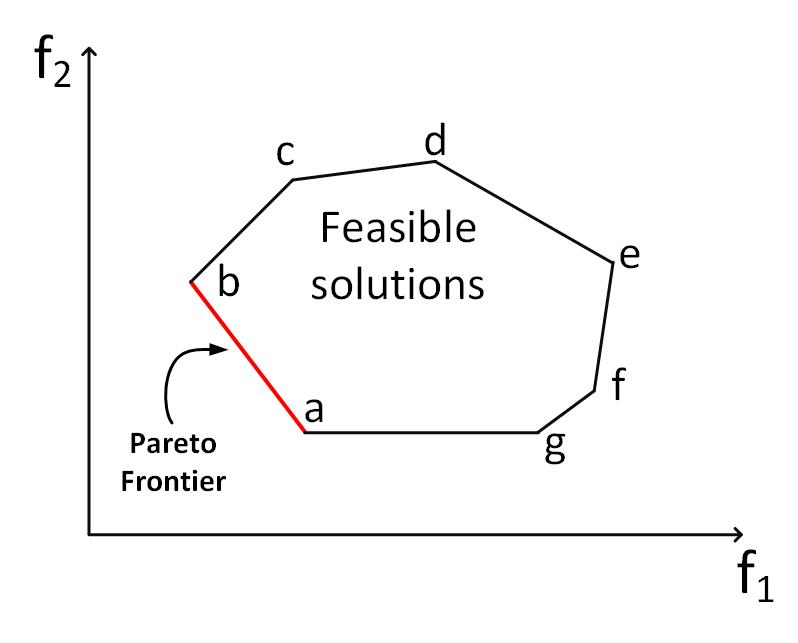}
	\vspace{-0.20in}
\caption{MOO with two objective functions}
\protect\label{fig:pareto}
\end{figure}

Figure~\ref{fig:pareto} shows a  MOO problem with two objective functions $f_1$ and $f_2$.  The boundary $\overleftrightarrow{ab}$ is the Pareto frontier  as it consists of all the Pareto optimal solutions. 

\subsection{Cooperative Game Theory}
Cooperative game theory provides a set of analytical tools that assist in understanding the behavior of rational players in a cooperative setting \cite{han2012game}. Players can have agreements among themselves that affect  strategies as well as utilities  obtained by the game players. Coalition game, a type of cooperative game, deals with the formation of  coalitions, namely groups of two or more cooperating players. Formally,

\begin{definition}{Coalition Game \cite{han2012game}:}\label{def:charac}
Any coalition game with \emph{non-transferable utility} (discussed below)  can be represented by the pair $(\mathcal{N},\mathcal{V})$ where $\mathcal{N}$ is the set of players that play the game, while $\mathcal{V}$ is a set of payoff vectors such that \cite{roger1991game}:

\begin{enumerate}
	\item $\mathcal{V}(S), \forall S \subseteq \mathcal{N}$ should be a closed and convex subset of $\mathbb{R}^{|\mathcal{N}|}$.
	\item $\mathcal{V}(S)$ should be comprehensive, i.e., if we are given payoffs $\mathbf{x} \in \mathcal{V}(S)$ and $\mathbf{y} \in \mathbb{R}^{|S|}$ where $\mathbf{y}\leq \mathbf{x}$, then $\boldsymbol{y} \in \mathcal{V}(S)$. In other words, if the members of coalition $S$ can achieve a payoff allocation $\mathbf{x}$, then the players can change their strategies to achieve an allocation $\mathbf{y}$ where $\mathbf{y}\leq \mathbf{x}$. 
	\item The set $\{\mathbf{x}| \mathbf{x} \in \mathcal{V}(S) \;$ and $\; x_n \geq z_n, \forall n \in S \}$, with $z_n = \max\{y_n|\mathbf{y} \in \mathcal{V}(\{n\}) \}\le \infty \; \forall n \in \mathcal{N}$ should be a bounded subset of $\mathbb{R}^{|S|}$. In other words, the set of vectors in $\mathcal{V}(S)$ for a coalition $S$ where the coalition members receives a payoff at least as good as working alone (non-cooperatively) is a bounded set. 
\end{enumerate}

\end{definition}

\begin{definition}{Utility:}
The payoff or gain a player receives in the game is known as utility.  
\end{definition}

\begin{definition}{Value of a coalition:}
The sum of utilities of the coalition members is known as the value ${v}$ of a coalition.
\end{definition}

\begin{definition}{Non-Transferable Utility (NTU)\cite{han2012game}:}
	If the total utility of any coalition cannot be  assigned a single real number  or if there is a rigid restriction on  utility distribution among players, then the game has a \emph{non-transferable} utility. In other words, the payoff each player receives and the value of coalition depend on the game strategy. Utility cannot be transfered freely among the players. 
	
\end{definition}

\begin{definition}{Characteristic function:}
	The characteristic function for any coalition game with NTU is a function that assigns  a set of payoff vectors, $\mathcal{V}(S) \subseteq \mathbb{R}^{|S|}$, to game players where each element of the payoff vector $x_n$ represents a payoff that player $n \in S$ obtains  given a strategy selected by the player $n$ in the coalition $S$.   
\end{definition}

\begin{definition}{Characteristic Form Coalition Games \cite{han2012game}:}
	A coalition game is said to be of \emph{characteristic} form, if the value of coalition $S \subseteq \mathcal{N}$ depends only on  the members of the coalition. (This means that players that are not part of the game don't impact the value of coalition.)
\end{definition}

\begin{definition}{Superadditivity of NTU games:}
A canonical game with NTU is said to be superadditive if the following property is satisfied. 
\begin{align}\label{eq:superadditivity}
  & \{x \in \mathbb{R}^{S_1\cup S_2}|(x_n)_{n \in S_1} \in \mathcal{V}(S_1), (x_m)_{m \in S_2}\nonumber \\
& \in \mathcal{V}(S_2)   \}   \subset \mathcal{V}(S_1 \cup S_2) \;  \forall S_1 \subset \mathcal{N}, S_2 \subset \mathcal{N}, S_1 \cap S_2 =\emptyset. 
\end{align}
where $x$ is a payoff allocation for the coalition  $S_1 \cup S_2$. 

\end{definition}

That is,  if any two disjoint coalitions $S_1$ and $S_2$ form a large coalition $S_1 \cup S_2$, then the coalition $S_1 \cup S_2$ can always give its members the payoff that they would have received in the disjoint coalition  $S_1$ and $S_2$. It is worth mentioning that $\mathcal{V}$ is the set of payoff vectors while ${v}$ is the sum of payoffs  for all coalition players. 
\begin{definition}{Canonical Game:}
A coalition game is canonical if it is superadditive and in characteristic form.
\end{definition}
The \emph{core} is a  widely used concept for dealing with canonical games as discussed below. 

\subsection{Core}
To help us explain the core, we first define some terms   \cite{han2012game}.

\begin{definition}{Group Rational:}\label{def:groupRational}
	A payoff vector $\textbf{x}\in \mathbb{R}^\mathcal{N}$ for distributing the coalition value $v(\mathcal{N})$ to all players is group-rational if $\sum_{n \in \mathcal{N}}x_n=v(\mathcal{N})$.
\end{definition}

\begin{definition}{Individual Rational:}\label{def:indRational}
	A payoff vector $\textbf{x}\in \mathbb{R}^\mathcal{N}$ is individually rational if every player can obtain a payoff no less than acting alone (not in a coalition), i.e., $x_n \geq v(\{n\}), \forall n\in \mathcal{N}$.
\end{definition}

\begin{definition}{Imputation:}
	A payoff vector that is both individual and group rational is  an imputation. 
\end{definition}

\begin{definition}{Grand Coalition:}
The coalition formed by all game players $\mathcal{N}$ is the grand coalition.
\end{definition}

Based on the above definitions, we can  define the core of an NTU canonical coalition game as:

\begin{definition}{Core\cite{han2012game}:}
	For any NTU canonical game $(\mathcal{N}, \mathcal{V})$, the core $C_{NTU}$ is the set of imputations in which no coalition $S\subset\mathcal{N}$ has any incentive to reject the proposed payoff allocation and deviate from the grand coalition to form the smaller  coalition $S$ instead. Mathematically, this is 
	\begin{align}\label{eq:coreTU}
	\mathcal{C}_{NTU}&= \{ \mathbf{x} \in \mathcal{V}(\mathcal{N}) | \forall S \subset \mathcal{N}, \nexists \mathbf{y} \in \mathcal{V}(S) \; such\; that \; y_n > x_n, \nonumber \\
	&  \forall n \in S    \}.
	\end{align}
\end{definition}
\begin{remark}\label{rem:paretocore}
	Any payoff allocation from the core is Pareto-optimal as evident from definition of the core. Furthermore, the grand coalition formed is stable, i.e., no two players will have an incentive to leave the grand coalition to form a smaller coalition. 
\end{remark}
However, the core is not always guaranteed to be non-empty. If the core is non-empty, it will be a convex set \cite{shapley1971cores}, whereas a convex set consists of infinite convex combinations of the points within the set \cite{boyd2004convex}. Therefore, a non-empty core  is large and finding a payoff allocation (vector) in the core is not easy.    
\begin{definition}{User/Application Request Satisfaction:}\label{def:rs}
	The extent to which a request for resources is satisfied.  Mathematically, if $r$ is the requested amount and $x$ is the provisioned amount, then application request satisfaction is given by $x/r$. 
\end{definition}

%% file: 02-sysmodel2.tex

 \begin{figure}
	\centering
	\includegraphics[width=0.48\textwidth]{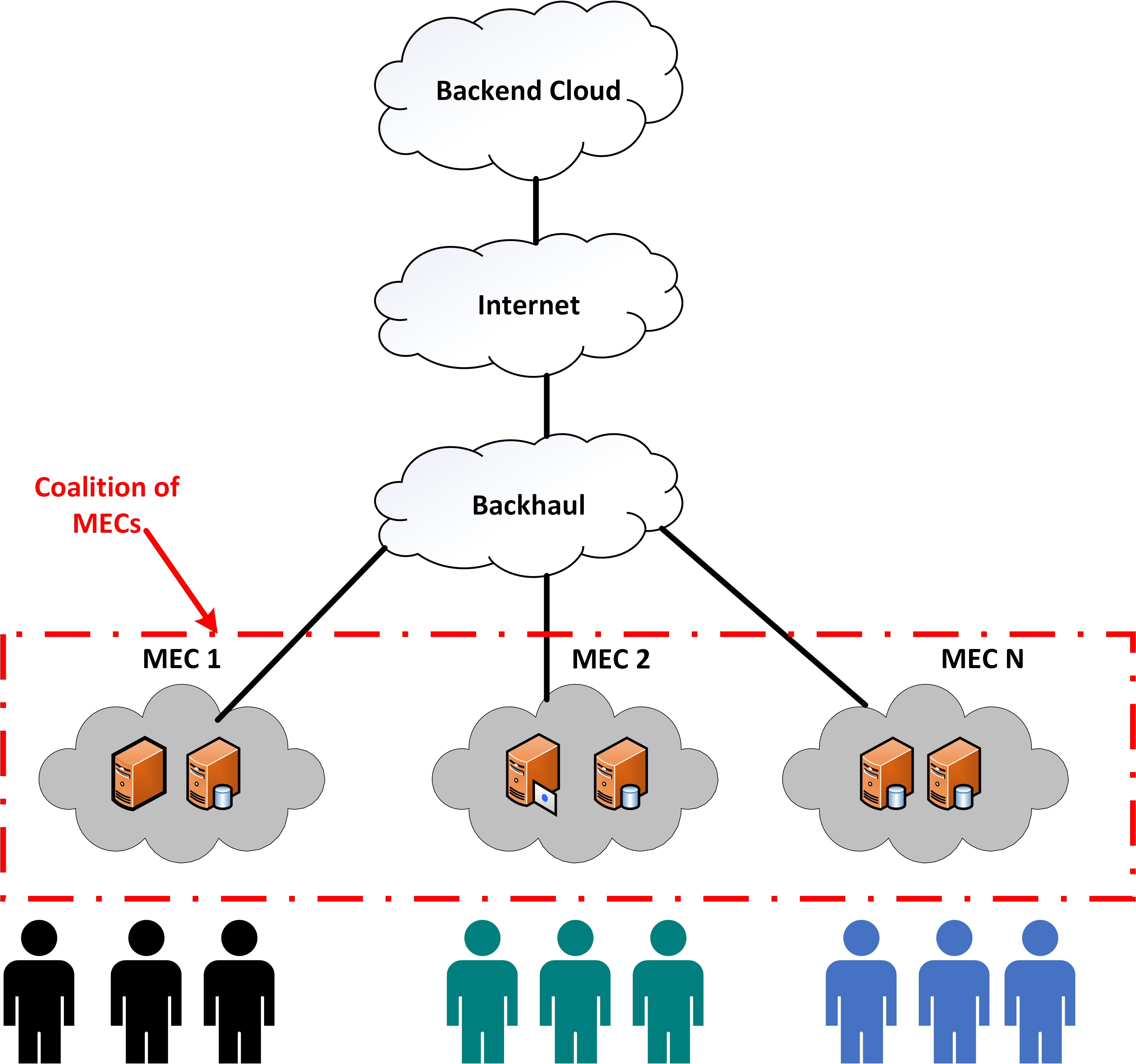}
	\caption{Our system model with multiple MEC providers}
	\protect\label{fig:system_model}
	\vspace{-0.2in}
\end{figure}

Let $\mathcal{N}=\{1,2, \cdots, N\}$ be the set of   edge cloud-service providers. 
We assume that each service provider has a set of $\mathcal{K}=\{1,2,\cdots, K\}$ different types of resources such as communication, computation and storage resources.  Let the available resources at service provider  $n$ be denoted by $C_{n} = \{ C_{n,1},\cdots, C_{n,K}\}$, where  $C_{n,k}$ is the amount of type $k$ resources available at  service provider $n$. 
  Each service provider $n$ has a set of  native (local) applications $\mathcal{M}_n= \{1,2,\cdots, M_n\}$. 
The set of all applications that ask for resources from all  service providers (i.e., the grand coalition of service providers) is given by $\mathcal{M}=\mathcal{M}_1\cup\mathcal{M}_2\cdots\cup \mathcal{M}_N , \;$ where $\mathcal{M}_i \cap \mathcal{M}_j=\emptyset, \; \forall i \neq j,$ i.e., each provider has an entirely different set of native applications. 
Every provider  $n \in \mathcal{N}$ has a request (requirement) matrix $R_{n}$.
\begin{equation}
\label{eq:R_Req}
R_{n}=\Biggl[\begin{smallmatrix}
\mathbf{r_{n}^1}\\ 
.\\ 
.\\ 
.\\
\mathbf{r_{n}^{{M}_n}}
\end{smallmatrix}\Biggr] = \Biggl[\begin{smallmatrix}
r_{n,1}^{1} & \cdots  &\cdots   & r^{1}_{n,{K}} \\ 
. & . &.  &. \\ 
. & . &.  &. \\
. & . &.  &. \\
r_{n,1}^{{M}_n}&\cdots  &\cdots   & r_{n,K}^{{M}_n}
\end{smallmatrix}\Biggr].
\end{equation}
where the $j^{th}$ row corresponds to  application $j$, while  columns represent different resources. That is,  $r_{n,k}^{j}$ is the amount of resource $k$ that application $j \in \mathcal{M}_n$ requests from provider $n$. When the service provider  works alone\footnote{Not borrowing from or renting out its resources to any other service provider.},  its objective is to maximize its utility by allocating  resources to its native applications and improving application satisfaction. A service provider $n$ earns a non-negative \emph{monotonically non-decreasing} utility $u_n^j (x_{n,k}^j)$ by making  allocation decision $x_{n,k}^j$, i.e., allocating $x_{n,k}^j$ amount of resource $k$ to application $j\in \mathcal{M}_n$, where  $\mathbf{x}_n^j= [x_{n,1}^j, x_{n,2}^j,\cdots, x_{n,K}^j]^T$. 
 Table \ref{tab:not} summarizes the notations used throughout the paper. Below we present the optimization formulation for a single service provider. 


\begin{table}[]
	\centering
	\caption{List of notations used throughout the paper}
	\begin{tabular}{|p{1cm}|p{6.9cm}|}
		\hline
		\textbf{Notation}	& \textbf{Description} \\ \hline
		$\mathcal{{N}}, N, n$	& Set, number and index of game players (service providers)  \\ \hline
		$\mathcal{K}, K, k$	& Set, number and index of resources \\ \hline
		$\mathcal{M}, M, j$	& Set, number and index of applications \\ \hline
		$\mathcal{M}_{n}$	& Set of native applications at player (provider) $n$  \\ \hline
		${S}$	& A coalition of game players (providers), where $S \subseteq \mathcal{{N}}$  \\ \hline
		$\mathcal{V}$	& Set of payoff vectors    \\ \hline
		${v}(S)$	& Value of coalition $S$    \\ \hline
		$\mathcal{G}_1$	& Set of players (providers) with resource deficit    \\ \hline
		$\mathcal{G}_2$	& Set of players (providers) with resource surplus   \\ \hline
		$C$	& Capacity vector of all game players (providers) \\ \hline
		$C_n$	& Capacity vector of player  (provider)$n$ \\ \hline
		$C_{n,k}$	& Capacity of resource $k$ at player (provider) $n$ \\ \hline
		$R_{n}$	& Request matrix at player (provider) $n$ \\ \hline
		$r_{n,k}^{j}$	& Request of application $j$ for resource $k$ from player (provider) $n$ \\ \hline
		$x_{n,k}^{j}$	& Allocation decision of resource $k$ for application $j$ at player (provider) $n$ \\ \hline
	$\mathbf{x}_{n}^{j}$	& Allocation decision vector for application $j$ at player (provider) $n$ when working alone \\ \hline

	$\mathbf{X}_n$ & Allocation decision for player (provider) $n$ in the  coalition \\ \hline 
	$\mathbf{X}$ & Allocation decision for the entire coalition  \\ \hline 
		$\mathcal{X}_n$ & Allocation decision for player (provider) $n$ in the  coalition excluding its native applications \\ \hline	
		$^{m}\mathbf{X}_n$ & Allocation decision for applications of player (provider) $m$ at player (provider) $n$ in the  coalition s \\ \hline		
			$u_{n}^j({x}_{n,k}^j)$	& Utility player (provider) $n$ earns  by allocating resource $k$ to application $j$ \\ \hline
		$u_{n}^j(\mathbf{x}_{n}^j)$	& Utility player (provider) $n$ earns  by allocating vector of  resources to application $j$ \\ \hline
		$GC$	& Grand coalition\\ \hline
		$RS$ & Request satisfiability\\ \hline
		$TU$ & Transferable utility\\ \hline
		$NTU$ & Non-transferable utility\\ \hline
		$CGT$ & Cooperative game theory\\ \hline
		$MOO$ & Multi-objective optimization\\ \hline
		$SOO$ & Single-objective optimization\\ \hline
		$MEC$ & Mobile edge cloud\\ \hline
	\end{tabular}
	\label{tab:not}
\end{table}
\subsection{Problem Formulation for Single Service Provider} \label{sec:prob}
 We  first present the resource allocation problem for a single, stand-alone provider (i.e., no resource sharing with other providers). 
For each provider  $n \in \mathcal{N}$, the allocation decisions  consist of the vectors $\mathbf{x}_n^1, \cdots, \mathbf{x}_n^{M_n}$. The optimization problem is: 
\begin{subequations}\label{eq:opt_single}
	\begin{align}
	\max_{\mathbf{x}_n^1, \cdots, \mathbf{x}_n^{M_n}}\quad& \sum_{j \in \mathcal{M}_n}f_n^j(\mathbf{x}_n^j), \label{eq:objsingle}\displaybreak[0]\\
	\text{s.t.}\quad & \sum_{j} x_{n,k}^{j}\leq C_{n,k}, \quad \forall k \in \mathcal{K} \label{eq:singlefirst} \displaybreak[0]\\
	& x_{n,k}^{j} \leq r^{j}_{n,k}, \quad \forall\; j\in \mathcal{M}_n, k \in \mathcal{K}, \label{eq:singlesecond} \displaybreak[1]\\
	&  x_{n,k}^{j} \geq 0, \quad \forall\; j\in \mathcal{M}_n, k \in \mathcal{K}. \label{eq:singlethird} \displaybreak[2]
	\end{align}
\end{subequations}
 The objective function $f_n^j(\mathbf{x}_n^j)$  is required to be a monotonic, non-negative and non-decreasing function. For illustration purposes in this paper, we assume that:
\begin{align}
f_n^j(\mathbf{x}_n^j) =w_{1}^ju_n^j(\mathbf{x}_n^j)+\sum_{k \in \mathcal{K}} \frac{x_{n,k}^{j}}{r_{n,k}^{j}}. \nonumber 
\end{align}
where the  term $ \sum_{k \in \mathcal{K}} \frac{x_{n,k}^{j}}{r_{n,k}^{j}}$  captures the request satisfaction (see Definition \ref{def:rs}) and $u_n^j(\mathbf{x}_n^j)$ is the non-decreasing and non-negative utility the  service provider earns by providing the resources to applications. 
The weight $w_{1}^j$ is  used for \emph{scaling} so that both the terms are on the same scale. 
The first constraint in \eqref{eq:singlefirst} indicates that the allocated resources cannot exceed capacity. The second constraint  in \eqref{eq:singlesecond} indicates that the allocated resources should not exceed the required amounts.  The last constraint, \eqref{eq:singlethird}  requires the allocation to be non-negative.

The goal of this single-objective optimization problem for every provider, as mentioned earlier,  is   to  maximize its utility by allocating available resources to its native applications and improving application satisfaction. We note that this single-provider formulation does not consider the following. For example, it is possible that
 a service provider $n$ may earn a larger utility by providing
 its resources to applications of other service providers. Furthermore, provider(s)
 may not have sufficient resources to satisfy requests of all its
 native applications, while other 
 providers  may have resource
 surpluses that can be ``rented" and utilized by other providers. In this work, we focus on strategies where each service
provider allocates resources to its native applications and
then shares any remaining  resources with other service
providers who are not able to fully satisfy their native applications.
These aspects are considered in the following formulation.

\subsection{Problem Formulation for Multiple Service Providers}
Let us assume that a mobile edge service provider $n$ does not have sufficient resources to satisfy all its native applications. Suppose that service provider $m$  has a resource surplus, after allocating its resources to its native applications that can  be shared with $n$. Allowing resource sharing among such providers can
 improve resource utilization and application satisfaction.
By sharing and allocating its type $k$ resources  to application $l$, service provider $m$ earns a non-decreasing and non-negative net utility $H_m^l(x_{m,k}^l)$, which in this paper is assumed to be 
\begin{align}\label{eq:1}
H_m^l( x_{m,k}^l)=w_1^l\Big(u_m^l (z_k^l+x_{m,k}^l)-u_m^l (z_k^l)  -D_m^l(x_{m,k}^l)\Big)\nonumber \\+ \Big(\frac{x_{m,k}^l}{r_{n,k}^l-z_k^l}\Big)^2, \forall l \in \mathcal{{M}}\backslash \mathcal{{M}}_m, n \in \mathcal{{N}}\backslash m.
\end{align} 
where $z_k^l= \sum_{n \in \mathcal{{N}}\backslash m}x_{n,k}^l$ is the amount of resource $k$ allocated already to application $l$ by providers other than $m$. $u_m^l (z_k^l+x_{m,k}^l)-u_m^l (z_{k}^l)$  captures the incremental increase in utility of provider $m$ due to the allocation of additional resource $x_{m,k}^l$.  $D_m^j(x_{m,k}^j)$ represents the communication cost of serving application $j \in  \mathcal{{M}}_n$ by using resource of type $k$ at  provider $m$ rather than its native service provider $n$. We assume that the communication cost is a positive monotonic non-decreasing function. The term $(\frac{x_{m,k}^l}{r_{n,k}^l-z_k^l})^2$, which takes  a value in  $[0, 1]$, is used to reduce resource fragmentation by providing a higher utility  for providing more resources at a particular provider.   

To consider each provider and its native applications, let  $F_n^{j}(\mathbf{x}_n^j)$ represent the utility  service provider $n$ earns by allocating resources to its native applications along with the increase in request satisfaction due to provision of resources to its applications at different service providers including itself.  
   We assume:
\begin{align}\label{eq:2}
F_n^{j}(\mathbf{x}_n^j) =w_1^ju_n^j(\mathbf{x}_n^j)+\sum_{k \in \mathcal{K}}\frac{\sum_{m \in \mathcal{{N}}}x_{m,k}^{j}}{r_{n,k}^{j}}.
\end{align}
Note that $F_n^{j}$ is  monotonic non-decreasing and non-negative. 
 Furthermore, we implicitly assume that $ u_n^j (\mathbf{x})=u_m^j (\mathbf{x}), \; \mathbf{x}\geq 0$.  The latter assumption reflects that different providers  earn the same utility when allocating the same amount of resources of the same type to a given application, either locally or remotely.  Furthermore, 

\begin{align}
\mathbf{x}^j= [\sum_{n \in \mathcal{{N}}}x_{n,1}^j, \sum_{n \in \mathcal{{N}}}x_{n,2}^j,\cdots, \sum_{n \in \mathcal{{N}}}x_{n,K}^j]^T. \nonumber
\end{align}
That is, the total resources allocated to any application $j \in \mathcal{M}$ is the sum of resource allocated to it across all service providers.  The resource sharing and allocation framework, based on resource requests and  capacities of service providers, has to make allocations that optimize the utilities of all the service providers $n \in \mathcal{N}$ while satisfying user requests as well.  The allocation decision is given by
$\mathbf{X}=\{\mathbf{X}_1, \mathbf{X}_2, \cdots, \mathbf{X}_N\}$, where $\mathbf{X}_n,\; \forall n \in \mathcal{N}$ is given by: 

\begin{equation}
\label{eq:X}
\mathbf{X}_{n}=\Biggl[\begin{smallmatrix}
\mathbf{x}^{1}_n\\ 
.\\ 
.\\ 
.\\
\mathbf{x}^{|\mathcal{{M}}|}_n
\end{smallmatrix}\Biggr] = \Biggl[\begin{smallmatrix}
x_{n,1}^{1} & \cdots  &\cdots   & x^{1}_{n,{K}} \\ 
. & . &.  &. \\ 
. & . &.  &. \\
. & . &.  &. \\
x_{n,1}^{|\mathcal{{M}}|}&\cdots  &\cdots   & x_{n,K}^{|\mathcal{{M}}|}
\end{smallmatrix}\Biggr].
\end{equation}

When service providers share resources,  each service provider aims to  maximize the sum of utilities obtained by a) allocating resources to its native applications;  b) allocating resources (if available) to  applications of other service providers;  c) improving the request satisfaction of its applications by using its
 own resources or borrowing from other service providers (when
needed and possible); and d) reducing resource fragmentation. 
In the resource sharing case, each service provider
allocates its resources by solving the following optimization problem:

\begin{subequations}\label{eq:opt_higher}
	\begin{align}
	\max_{\mathbf{X}_n}\quad&\sum_{j \in \mathcal{M}_n} F_n^{j}(\mathbf{x}_n^j)+\sum_{l \in \{\mathcal{M}\backslash \mathcal{M}_n\}} H_n^l (\mathbf{x}_n^l),\quad \forall n \in \mathcal{{N}}, \label{eq:obj}\\
	\text{s.t.}\quad 
	& \sum_{j} x_{n,k}^{j}\leq C_{n,k}, \quad \forall\; j\in \mathcal{M}, k \in \mathcal{K}, n \in \mathcal{N}, \label{eq:obj1} \displaybreak[0]\\
	& \sum_{m \in \mathcal{N}}x_{m,k}^{j} \leq r_{n,k}^{j}, \quad \forall\; j\in \mathcal{M}, k \in \mathcal{K}, n \in \mathcal{N},  \label{eq:obj2}\displaybreak[1]\\
	&  x_{n,k}^{j} \geq 0, \quad \forall\; j\in \mathcal{M}, k \in \mathcal{K}, n \in \mathcal{N}, \label{eq:obj3} \displaybreak[2] \\
	&  u_n^l(z_k^l+ {x}_{n,k}^l) - u_n^j({z}_{k}^l)\geq D_n^l({x}_{n,k}^l),\quad \nonumber \\ & \forall l \in \mathcal{M}\backslash\mathcal{{M}}_n, k \in \mathcal{K}, n \in \mathcal{N} \label{eq:obj4}, 
	\end{align}
\end{subequations}
where  $F_n^j(\mathbf{x}_n^j)$ and $H_n^l(\mathbf{x}_n^l)$ are given in \eqref{eq:1} and \eqref{eq:2}. 
Constraints \eqref{eq:obj1}-\eqref{eq:obj3} are similar to those for the formulation in \eqref{eq:opt_single}, except that \eqref{eq:obj4} indicates that the utility earned by sharing resources must be higher than the communication cost.

As the utility $u_n^j,\; \forall n\in \mathcal{{N}} $ can differ for each application $j$, \eqref{eq:opt_higher} for all providers $n$ represents   a multi-objective optimization problem, which will be solved by a game-theoretic approach as follows.

%% file: 03-opt-problem.tex
In this section, we first present general properties of our game and show that the resource-sharing problem for multiple providers  can be modeled as a canonical cooperative game with non-transferable utility. We then show that the core is non-empty, by proving that our canonical game is \emph{cardinally convex}\cite{sharkey1981convex}. Finally, we propose algorithms in Section \ref{sec:algorithms} to obtain allocations from the core.

\subsection{General Game Properties}\label{sec:generalgame}

We model each service provider as a player in our game to obtain the optimal resource sharing/allocation decision. 
Let $\mathcal{N}$ be the set of players that  play the resource sharing and allocation game. The \emph{value of coalition} for the game players $S \subseteq \mathcal{N}$ is given by, \eqref{eq:payoff_function}, where $\mathcal{F}_S$ is the feasible set for resource sharing and allocation given by, \eqref{eq:obj1}-\eqref{eq:obj3} and $\mathcal{M}_S$ is the set of applications in $S$, 
\vspace{-0.1in}
\begin{align} \label{eq:payoff_function}
v({S})&=\sum_{{\substack{{n \in S}\\ {\mathcal{X} \in \mathcal{F}_S}}}}  
\bigg(\sum_{j \in \mathcal{M}_n}\big(w_1^ju_n^j(\mathbf{x}_n^j)+\sum_{k \in \mathcal{K}}\frac{\sum_{m \in {{S}}}x_{m,k}^{j}}{r_{n,k}^{j}}\big) +\nonumber \\
&  \Big(\sum_{l \in \{\mathcal{M}\backslash \mathcal{M}_n\}} w_1^l\big(u_n^l(\mathbf{z}^l+\mathbf{x}_n^l)-u_n^l(\mathbf{z}^l)-D_n^l(\mathbf{x}_n^l)\big)+\nonumber \\
&\sum_{k \in \mathcal{K}}\Big(\frac{x_{n,k}^l}{r_{p(l),k}^l-z_k^l}\Big)^2\Big)\bigg).
\end{align}
Here $p(l)$ indicates the native player to which application $l$ belongs. Our strategy  to achieve the aforementioned value of coalition is that each player first allocates resources to its own applications and then shares its remaining resources (if any) with  other players. 

\begin{remark}
	
A	resource allocation and sharing problem (with multiple objectives) for the aforementioned system model  can be modeled as a canonical cooperative game with NTU. 
\end{remark}
This is so because the stated problem satisfies the following two
conditions.

\begin{itemize}
	\item \textbf{Characteristic form of payoff:} As the utility function in our  formulated resource sharing and allocation problem  only relies on  service providers that are part of the coalition, the game or payoff function is of characteristic form.
	\item \textbf{Superadditivity:} 
	Let $S_1, S_2 \subseteq \mathcal{N}$ where $S_1, S_2$ are non-empty and $S_1 \cap S_2 =\emptyset$. Hence, $S_1,\; S_2 \subset (S_1 \cup S_2)$. Superadditivity  follows from the monotonic utility functions. 
	
\end{itemize}

\begin{definition}[Cardinally Convex Games\cite{sharkey1981convex}]
	An NTU game is said to be cardinally  convex if $\forall\; S_1,\; S_2 \subseteq \mathcal{{N}}$:
	\begin{equation}\label{eq:convexgames}
\mathcal{V}(S_1)\cup \mathcal{V}(S_2) \subseteq \mathcal{V}(S_1 \cup S_2) \cup\mathcal{V}(S_1 \cap S_2). 
\end{equation} 	
\end{definition}
whereas $\mathcal{V}$ is given in  definition \ref{def:charac}.
\begin{theorem}\label{thm:convex}
	Our canonical game is cardinally convex. 
\end{theorem}
\begin{proof}
	Let us consider two coalitions $S_1$ and $S_2$ that are non-empty subsets of $\mathcal{N}$. Then from super-additivity of the game,    
	\begin{equation}
	\label{eq:superADD}
	\mathcal{V}(S_1)\cup\mathcal{V}(S_2) \subset \mathcal{V}(S_1 \cup S_2).
	\end{equation}
Furthermore, due to the non-negative and non-decreasing nature of the utility used,
		\begin{equation}
	\label{eq:superADD2}
 \mathcal{V}(S_1 \cup S_2) \subseteq \mathcal{V}(S_1 \cup S_2) \cup\mathcal{V}(S_1 \cap S_2).
	\end{equation}
The proof follows. 	
\end{proof}

\begin{remark}
	The core of any convex game $(\mathcal{N}, \mathcal{V})$ is non-empty\cite{sharkey1982cooperative} and large \cite{maschler2013}. 
\end{remark}

\begin{remark}\label{lemma:MOOGame}
	Our canonical cooperative game $(\mathcal{N}, \mathcal{V})$ with NTU can be used to obtain the Pareto-optimal solutions for the multi-objective resource  allocation for various service providers.
\end{remark}
The convexity of our game proves that the core is non-empty. However, obtaining an allocation from the core is challenging, which  we address in the following subsection.

\subsection{Proposed Algorithms}\label{sec:algorithms}

We first propose an efficient algorithm that provides  an allocation from the core. Then, we propose a second algorithm based on polygamous matching \cite{maschler2013} that obtains an allocation from the core subject to  resource fragmentation constraints  at the cost of higher computational complexity when compared with the first algorithm. Below we discuss these algorithms in detail. 

 
\subsubsection{Game-theoretic Pareto Optimal Allocation}
  We propose a  computationally efficient algorithm,  Algorithm \ref{algo:alg2}, that obtains a Pareto-optimal allocation by solving $N+|\mathcal{{G}}_2|$ optimization problems, where $\mathcal{G}_2$ represents the set of players with resource surplus. Inputs are available resources/capacities, application requests and  utilities of all game players. The algorithm outputs an allocation. In Step 1,  utilities of  players, the allocation decision,  and vectors $\boldsymbol{v}$, $\boldsymbol{A}$ and $\boldsymbol{B}$ are initialized. $\boldsymbol{v}$ stores the utility each service provider achieves when working alone.  $\boldsymbol{A}$  stores the utility of players in $\mathcal{G}_2$ whereas  $\boldsymbol{B}$ stores the utility increase due to request satisfaction of applications in resource deficit group $\mathcal{G}_1$. 
     Step 2 calculates every player's utility in the absence of resource sharing, i.e., the single objective optimization problem in \eqref{eq:opt_single} is solved by every player. 
  To consider the resources allocated in Step 2, the remaining resource capacities and requests for various applications $C, R$ are updated and stored in $C'$ and $R'$, respectively.
  Step $4$ divides players into  two groups $\mathcal{G}_1$ (resource deficit) and $\mathcal{G}_2$ (resource surplus), on the basis of updated capacities and requests. Step $5$ allocates the shared resources of players in $\mathcal{G}_2$ to players in $\mathcal{G}_1$. 
    $\boldsymbol{z}^j$  denotes the total resources
  allocated to an application $j$ previously in Step $2$ and up to the current iteration of Step
 $5$. Up to this point, a given application $j$ can receive resources from multiple players due to resource sharing.
   In step $6$, the utility reflecting an  increase in application satisfaction for native applications of players in $\mathcal{G}_1$ is calculated using the allocation obtained in Step $5$. 
   Below, we prove that Algorithm \ref{algo:alg2} provides an allocation from the core. 
\begin{algorithm}
	\begin{algorithmic}[]
		\State \textbf{Input}: $R, C,$ and vector of players' utility function  $\mathbf{u}$ 
		\State \textbf{Output}: $\mathbf{X}$ 
		\State \textbf{Step $1$:}  $ \mathbf{u}(\mathbf{X}) \leftarrow$0,  $ \mathcal{X} \leftarrow$0,  $ \boldsymbol{v} \leftarrow$0, $\boldsymbol{A} \leftarrow$0, $\boldsymbol{B} \leftarrow$0, $\boldsymbol{z} \leftarrow$0 
		\State \textbf{Step $2$:}
		\For{\texttt{$n \in \mathcal{N}$}}
		\State $ {v(\{n\})}\hspace*{-0.1cm}\leftarrow$\texttt{Objective function value  of   \eqref{eq:opt_single}} 
		\State $  \mathbf{x}_n^1, \cdots, \mathbf{x}_n^{M_n}\leftarrow$\texttt{Allocation decision of \eqref{eq:opt_single}} 
		\State {Update $z^j$'s based on the allocation  decision of \eqref{eq:opt_single} } 
		\EndFor
		\State \textbf{Step $3$:}  Update $C$ and $R$ based on Step 2 \\  $\quad \quad \quad \;  C'\leftarrow$ $C_{updated}$,$\; R'\leftarrow$ $R_{updated}$
		\State \textbf{Step $4$:} Divide the players into two subsets $\mathcal{G}_1$ and $\mathcal{G}_2$ representing players  with resource deficit and resource surplus
		\State \textbf{Step $5$:}
		\For{\texttt{$n \in \mathcal{G}_2$}}
		\State 
		$\mathcal{X}_n\hspace*{-0.1cm}=\hspace*{-0.1cm}\mathbf{X}_n\backslash\{\mathbf{x}_n^1, \cdots, \mathbf{x}_n^{M_n}\}\hspace*{-0.1cm}\leftarrow$
		\texttt{Optimal allocation decision of  \eqref{eq:opt_single_j} $\forall m \in \mathcal{{G}}_1$ }
    \State 		 {$ {A_{n}^j}\hspace*{-0.1cm} \leftarrow$ \texttt{Utility earned by $n \in \mathcal{{G}}_2$ due to resource allocation $\mathcal{X}_n$ }}
	\State {Update } $z^j${'s based on $\mathcal{X}_n$}
		\State Update $C'$ and $R'$
		\EndFor
		\State \textbf{Step $6$:}
		\For{\texttt{$m \in \mathcal{G}_1$}}
		\State $ {B_{m}} \leftarrow$ \texttt{Utility earned by  satisfying users $j\in\mathcal{M}_m$  due to allocation decision in  Step $5$}
		\EndFor
		
	\end{algorithmic}
	\caption{Game-theoretic Pareto optimal allocation (GPOA)}
	\label{algo:alg2}
\end{algorithm}

\begin{theorem}\label{thm:alg1}
	The solution obtained from Algorithm \ref{algo:alg2}  lies in the \emph{core}.
\end{theorem}
\begin{proof}
	To prove the theorem, we need to show that a)  utilities obtained using Algorithm \ref{algo:alg2}   are individually rational and  group rational, and b) no group of players  has the incentive to leave the grand coalition to form another sub-coalition $S \subset \mathcal{N}$. 
	\newline \indent\textbf{Individual Rationality:} (See definition \ref{def:indRational}) For each player $n \in \mathcal{N}$,  $v(\{n\})$ obtained by solving   \eqref{eq:opt_single} is the utility a player  obtains by working alone in the absence of resource sharing. Because the utilities are non-negative, $A_n\geq0,\;\forall n \in \mathcal{{N}}$.
	 Furthermore utilities of  players in $\mathcal{G}_1$ may increase due to  increase in request satisfaction if its applications are provided additional resources by players in  $\mathcal{G}_2$ given by $B_{n}$. 
	 Hence  Algorithm \ref{algo:alg2} produces an  individually rational resource allocation.
	\newline \indent\textbf{Group Rationality:} The value of the grand coalition $v\{\mathcal{N}\}$ as per Equation \eqref{eq:payoff_function} is the sum of different utilities. 
	Steps $2$, $5$ and $6$ of Algorithm \ref{algo:alg2} obtain the sum of utilities as well. Hence the solution obtained as a result of Algorithm \ref{algo:alg2} is group rational. 
	\par Furthermore, due to super-additivity of the game and monotonic non-decreasing nature of the utilities, no group of players has the incentive to form a smaller coalition. Hence Algorithm \ref{algo:alg2} provides a solution from the core.  
\end{proof}

\begin{remark}\label{rem:distributed}

Algorithm \ref{algo:alg2} is a distributed algorithm. All players first allocate resources to their native applications and then update their capacities and requests. The updated requests and capacities are broadcasted by all players to obtain $\mathcal{{G}}_1$ and $\mathcal{{G}}_2$. The order to execute Step $5$ is chosen based on some criteria (see details below).  Players in $\mathcal{{G}}_2$, after allocating resources, send the updated requests  to other players in  $\mathcal{{G}}_2$  that have not yet executed Step $5$.

\end{remark}
\begin{remark}\label{rem:paretooptimality}
	Any payoff allocation from the core  generated by Algorithm 1 is Pareto-optimal.  
\end{remark}
\begin{remark}\label{rem:forlooporder}
	The allocation decision obtained using Algorithm \ref{algo:alg2} always belongs to the core irrespective of the order in which players execute Step $5$. 
\end{remark}\vspace{-0.3in}
\begin{subequations}\label{eq:opt_single_j}
	\begin{align}
	\max_{\mathcal{X}_n}\quad& \sum_{m\in \mathcal{G}_1}\Bigg(\sum_{j \in \mathcal{M}_m}\bigg(w_1^j\big(u_n^j({\boldsymbol{z}}^j+\mathbf{x}_n^j)-u_n^j({\boldsymbol{z}}^j)-\nonumber \displaybreak[0]\\
	&D_n^j(\mathbf{x}_n^j)\big)+\sum_{k \in \mathcal{K}}\Big(\frac{x_{n,k}^j}{r_{m,k}^j-z_k^j}\Big)^2\bigg)\Bigg),\;  \forall n \in \mathcal{G}_2 \label{eq:objsinglej}\\
	\text{s.t.}\quad & \sum_{j} x_{n,k}^{j}\leq C'_{n,k}, \; \forall k \in \mathcal{K},  \forall j \in \mathcal{M}_m,  \forall m \in \mathcal{{G}}_1, \label{eq:singlefirstj} \displaybreak[0]\\
	& x_{n,k}^{j} \leq r^{'j}_{m,k}, \quad k \in \mathcal{K},  \forall j \in \mathcal{M}_m,  \forall m \in \mathcal{{G}}_1, \label{eq:singlesecondj} \displaybreak[1]\\
	&  x_{n,k}^{j} \geq 0, \quad  k \in \mathcal{K},  \forall j \in \mathcal{M}_m,  \forall m \in \mathcal{{G}}_1, \label{eq:singlethirdj} \displaybreak[2]\\
	& u_n^j({{z}}_k^j+{x}_{n,k}^j)-u_n^j({{z}}_k^j) \geq D_n^j({x}_{n,k}^j), \;  k \in \mathcal{K},   \nonumber \\
	& \forall j \in \mathcal{M}_m, \forall m \in \mathcal{{G}}_1. \label{eq:singlefourthj}
	\end{align}
\end{subequations}

Based on Remark \ref{rem:forlooporder}, it is  important to note that Algorithm \ref{algo:alg2} can generate different solutions according to the order in which players execute Step $5$. Some candidate ordering schemes listed below. 

\begin{enumerate}
    \item \textbf{Capacity ascending order (CAO):} Players in $\mathcal{G}_2$ are arranged in an ascending order on the basis of the remaining capacities of a resource $k$, which is then used for executing  Step 5.  
	\item \textbf{Capacity  descending order (CDO):} Players in $\mathcal{G}_2$ are arranged in a descending order on the basis of the remaining capacities of a resource $k$, which is then used for executing  Step 5.  
	
	\item \textbf{Random Order:} The order in which players execute  Step $5$ is random. 
\end{enumerate}

While all these variants provide allocations from the core, finding the most beneficial order for the \texttt{for-loop} in Step $5$ requires evaluating all possible combinations of possible orders for executing the for-loop. 
Since Algorithm \ref{algo:alg2} always provides an allocation from the core,
it is desirable to identify the order of surplus players considered in Step $5$ that also considers physical constraints such as
resource fragmentation across multiple players (service providers).
Toward this goal, we present below a matching-based resource
sharing algorithm that further reduces resource fragmentation while
providing an allocation from the core.

\subsubsection{Polyandrous-Polygamous Matching for Resource Sharing}
Our proposed algorithm is based on polyandrous polygamous matching \cite{baiou2000many} that matches players (service providers) in $\mathcal{{G}}_1$ to $\mathcal{{G}}_2$ with the goal of maximizing utility for players in  $\mathcal{{G}}_2$ and reducing  resource fragmentation for supporting native applications from players  in $\mathcal{{G}}_1$. The idea behind our matching algorithm, inspired from \emph{two-sided markets} \cite{maschler2013}, for reducing resource fragmentation is as follows. 
Players in $\mathcal{{G}}_1$ have specific preferences for which  players in $\mathcal{{G}}_2$ they would like to obtain resources from. 
 On the other hand, players in $\mathcal{{G}}_2$ aim to maximize their utilities. As our framework relies on  monotonic non-decreasing utilities, the objectives of players in  $\mathcal{{G}}_1$ and $\mathcal{{G}}_2$ can be mapped to each other, i.e., the more  resources players in $\mathcal{{G}}_2$ share with players in $\mathcal{{G}}_1$, the higher  will their utilities  be and the lower  the resource fragmentation will be. Therefore, during each round, resources of one player in $\mathcal{{G}}_2$ are assigned to applications of a player in $\mathcal{{G}}_1$ that produces the largest increase in its utility. This  means that a $\mathcal{{G}}_2$ player will provide as many resources as possible (increasing its utility) that in turn should reduce resource fragmentation for applications in $\mathcal{{G}}_1$. From resource fragmentation perspective, our scheme will transform into \emph{polygamous}/many-to-one matching \cite{maschler2013} or \emph{stable marriage}/one-to-one matching\cite{gale1962college} in the best case scenario.   
Below, we discuss the different steps of the algorithm. 
 \par 
In Algorithm \ref{algo:algPoly}, all  variables are initialized in Step $1$.  $\boldsymbol{z}$ is a vector that stores  allocations for  every application $j \in \mathcal{{M}}$ across different players.  In Step $2$, all game players $n \in \mathcal{N}$ allocate resources to their native applications. Request matrices and capacity vectors are updated in Step $3$, and  players are divided into two groups $\mathcal{{G}}_1$ and $\mathcal{{G}}_2$ in Step $4$.  Here $\boldsymbol{A}$ and $\boldsymbol{B}$ records utilities of   players in $\mathcal{{G}}_2$ and $\mathcal{{G}}_1$, respectively. Step $5$ is  the core of  the algorithm. In Step $5a$, Algorithm \ref{algo:matrixConstruction} is used to construct a matching matrix $\mathcal{J}$. 
The elements of matching matrix $\mathcal{J}$  ($m$ rows and $n$ columns) are obtained by solving and obtaining the objective function value of the the problem in \eqref{eq:opt_matching} for every possible pair of $(m\in\mathcal{G}_1, n \in \mathcal{{G}}_2)$. \eqref{eq:opt_matching} maximizes the incremental utility a player in $\mathcal{{G}}_2$ earns by sharing its resources with a player in $\mathcal{{G}}_1$, whereas the second summation term is used to reduce resource fragmentation. 
\eqref{eq:matchingfirst} indicates that the allocated resources cannot exceed the capacity. \eqref{eq:matchingsecond} indicates that the allocation cannot be more than the required resources. \eqref{eq:matchingthird} shows that the allocation decision is non-negative whereas \eqref{eq:matchingfourth} indicates that the incremental increase in utility cannot be less than the communication cost.


Algorithm \ref{algo:matrixConstruction} also provides  total amount of resources used for each element of  $\mathcal{J}$  and allocation decision $^{m}X_n$ that captures how resources are allocated to all applications $j \in \mathcal{{M}}_m$ for $m \in \mathcal{{G}}_1$ at player $n \in \mathcal{{G}}_2$. 
In $5b$, the largest element in $\mathcal{J}$, given by $\mathcal{J}_{m,n}$,  is used to match $m$ to $n$, i.e., the player $n$ (with resource surpluses) shares its resources with the player $m$ (with resource deficits).
If there are multiple maximum values, then the one that uses smallest number of resources is used to choose the matching.  In $5c$, applications $j \in \mathcal{{M}}_m$ are assigned to $n \in \mathcal{{G}}_2$. In $5d$, $\boldsymbol{A}_n$ is updated with the  utility of player $n \in \mathcal{G}_2$ based on the assignment in $5c$. Similarly,  $\boldsymbol{B}_m, \mathcal{X}_n$ and $\boldsymbol{z}^j$ are updated to reflect the increase in application satisfaction,  resources allocated by $n$ to different applications and resources allocated to $j \in \mathcal{{M}}_m$, respectively. Requests and capacities are updated in $5e$ whereas $5f$ updates elements of $\mathcal{{G}}_1$ and $\mathcal{{G}}_2$ if required. 



\begin{definition}{Objection to a matching\cite{maschler2013}:}
	A player $m \in \mathcal{{G}}_1$ and  $n \in \mathcal{{G}}_2$ object to a matching $M$ if they both  prefer being matched to each other than to whom  they are matched by $M$. 
\end{definition}

\begin{definition}{Stable matching\cite{maschler2013}:}
A matching $M$ is stable if there is no pair $m \in \mathcal{{G}}_1$ and   $n \in \mathcal{{G}}_2$ that objects to a matching. 
\end{definition}

\begin{theorem}
	The polyandrous-polygamous matching in Algorithm \ref{algo:algPoly} is stable. 
\end{theorem}
\begin{proof}
	We prove this by contradiction. Assume that the matching $M$ obtained from Algorithm \ref{algo:algPoly} is not stable and there exists  a pair $(m \in \mathcal{{G}}_1, n \in \mathcal{{G}}_2)$ that objects to the matching. $M$ has originally matched $m \in \mathcal{{G}}_1$ to $n'$ and $n \in \mathcal{{G}}_2$ to $m'$. 
	The objection by $n$ means that it can achieve a higher utility by being matched with $m$ when compared with the current match $m' \in \mathcal{{G}}_1$. However, this is not possible as  $m'$ was matched with $n$ based on the maximization problem given in \eqref{eq:opt_matching}, i.e., $m'$ provides  the highest utility to $n$ for its resources. 
	Hence, our assumption is wrong and the matching is stable. 
	

\end{proof}
\begin{remark}\label{rem:alg2}
	Algorithm \ref{algo:algPoly} provides an allocation from the core.
\end{remark}
\vspace{-0.1in}
\begin{remark}
The  optimization problems in \eqref{eq:opt_single}, \eqref{eq:opt_single_j} and \eqref{eq:opt_matching} are non-convex. The hardness of these problems depend on the parameters such as the utility functions, requests and capacities. Therefore, it is difficult to evaluate the hardness of these problems, but the instances we consider  in Section \ref{sec:exp_results} are easily solved using existing solvers.
In general, certain  non-convex optimization problems can usually be solved to  optimality by proving that the \emph{strong duality} holds \cite{boyd2004convex},  and then solving the convex dual of these non-convex problems. Detailed discussions on strong duality can be found in \cite{boyd2004convex} whereas \cite{tychogiorgos2013non} presents the necessary and sufficient conditions for strong duality of non-convex problems. If the optimal solutions are not found for  \eqref{eq:opt_single}, \eqref{eq:opt_single_j} or \eqref{eq:opt_matching}, then the allocation will not be from the core. 
\end{remark}

\begin{algorithm}
	\begin{algorithmic}[]
		\State \textbf{Input}: $R, C,$ and vector of players' utility function  $\mathbf{u}$ 
		\State \textbf{Output}: $\mathbf{X},$  $\mathbf{u}(\mathbf{X})$
		\State \textbf{Step $1$:}  $ \mathbf{u}(\mathbf{X}) \leftarrow0$,  $ \mathbf{X} \leftarrow0$,  $ \boldsymbol{v} \leftarrow0$, $[\mathcal{J}]_{|\mathcal{{G}}_1|\times |\mathcal{{G}}_2|}\leftarrow\mathbf{0}, \boldsymbol{z}\leftarrow\mathbf{0}$  
		\State \textbf{Step $2$:}
		\For{\texttt{$n \in \mathcal{N}$}}
		\State $ {v(\{n\})}\leftarrow$\texttt{Objective function value  from  Equation \eqref{eq:opt_single}} 
		\State $ \mathbf{x}_n^1, \cdots, \mathbf{x}_n^{M_n}\hspace*{-0.12in}\leftarrow$\texttt{Allocation decision from \eqref{eq:opt_single}} 
			\State {Update $z^j$'s based on the allocation  decision of \eqref{eq:opt_single} } 
		\EndFor
		\State \textbf{Step $3$:}  Update $C$ and $R$ based on Step 2 \\  $\quad \quad \quad \;  C'\leftarrow$ $C_{updated}$,$\; R'\leftarrow$ $R_{updated}$
		\State \textbf{Step $4$:} Divide the players into two subsets $\mathcal{G}_1$  and \hspace*{0.5in} $\mathcal{G}_2$ representing players  with resource deficit and \hspace*{0.55in}resource surplus \\$\quad\quad  \quad \quad \boldsymbol{A} \leftarrow0$, $\boldsymbol{B} \leftarrow0$
		\State \textbf{Step $5$:} 
		\NoDo
		\While  $\quad \mathcal{{G}}_1\neq \emptyset\quad  ||\quad \mathcal{{G}}_2 \neq \emptyset$
		\State
		\textbf{Step $5a$: }{$[\mathcal{J}]_{|\mathcal{{G}}_1|\times |\mathcal{{G}}_2|}\leftarrow$  Algorithm \ref{algo:matrixConstruction} for constructing\\ matching matrix}
	\State \textbf{Step $5b$: }
		\If  {Multiple maximum values in $[ \mathcal{J}]$}\\
		 {\hspace*{0.35in} Choose one with the lowest $\mathcal{R}_{m,n}$}
		\Else{}\\
	 { \hspace*{0.4in}Obtain the row $m$ and column $n$ of the  maximum \hspace*{0.4in} value in  $[\mathcal{J}]$}
		\EndIf
		\State {\textbf{Step $5c$: }Assign all $j \in \mathcal{{M}}_m, m \in \mathcal{{G}}_1$ to $n \in \mathcal{{G}}_2$ }
		\State {\textbf{Step $5d$: }Update $\boldsymbol{A}_n, \boldsymbol{B}_m, \mathcal{X}_n $ and $\boldsymbol{z}^j$ using the preceding assignment}
		\State\textbf{Step $5e$: } Update $C_n$  and $R_m$
		\State \textbf{Step $5f$: }
		\If  \texttt{ $C_n=0,$}\\
		\quad \quad \quad {$\;\mathcal{{G}}_2=\mathcal{{G}}_2\backslash n$}
		\EndIf
			\If  \texttt{ $R_m=0,$}\\
		\quad \quad \quad {$\;\mathcal{{G}}_1=\mathcal{{G}}_1\backslash m$}
		\EndIf
		\EndWhile
	\end{algorithmic}
	\caption{Polyandrous-Polygamous Matching based Pareto Optimal Allocation (PPMPOA)}
	\label{algo:algPoly}
\end{algorithm}

\begin{subequations}\label{eq:opt_matching}
	\begin{align}
	\max_{^{m}\mathbf{X}_n}\quad& \sum_{j \in \mathcal{M}_m}\Big(w_1^j\big(u_n^j({\boldsymbol{z}}^j+\mathbf{x}_n^j)-u_n^j({\boldsymbol{z}}^j)-D_n^j(\mathbf{x}_n^j)\big)+\nonumber \\
	& \sum_{k \in \mathcal{K}}\Big(\frac{x_{n,k}^j}{r_{m, k}^j-{z}_k^j}\Big)^2\Big),
\label{eq:objmatching}\\
	\text{s.t.}\quad & \sum_{j} x_{n,k}^{j}\leq C'_{n,k}, \; \forall k \in \mathcal{K},  \forall j \in \mathcal{M}_m,  \label{eq:matchingfirst} \displaybreak[0]\\
	& x_{n,k}^{j} \leq r^{'j}_{m,k}, \quad k \in \mathcal{K},  \forall j \in \mathcal{M}_m, \label{eq:matchingsecond} \displaybreak[1]\\
	&  x_{n,k}^{j} \geq 0, \quad  k \in \mathcal{K},  \forall j \in \mathcal{M}_m \label{eq:matchingthird} \displaybreak[2]\\
	& u_n^j({{z}}_k^j+{x}_{n,k}^j)-u_n^j({{z}}_k^j) \geq D_n^j({x}_{n,k}^j), \;  k \in \mathcal{K},  \nonumber \\
	&\forall j \in \mathcal{M}_m. \displaybreak[2] \label{eq:matchingfourth}
	\end{align}
\end{subequations}

\begin{lemma}
	Our resource sharing framework improves request satisfaction. 
\end{lemma}
\begin{proof}
	Let the average application request satisfaction for the non-resource sharing case be given by $RS_{n,1}$, where $RS_{n,1}$ is obtained following the resource allocation decision in Step $2$ of Algorithms \ref{algo:alg2} and \ref{algo:algPoly}, $\forall n \in \mathcal{{N}}$. As players in $\mathcal{{G}}_2$ share their resources with 
	players in $\mathcal{{G}}_1$,  average user satisfaction as a result of the allocation,  given by $RS_{n,2}$ is also  positive, i.e., $>0$. Hence,  user satisfaction achieved using our framework is given by $RS_{n}=RS_{n,1}+RS_{n,2}$, where $RS_{n}\ge RS_{n,1}$.  
\end{proof}
\subsection{Strategy-Proof Incentive}
In this section, we answer the question: \emph{Is there an incentive for a single player or group of players to cheat and report incorrect capacities and application requests  to other players?} That is, will misreporting the capacities or application requests improve the payoff a player or group of players receive? 
The answer is no as given by the following Theorem. 
\begin{theorem}
	When the service providers use Algorithms \ref{algo:alg2} and \ref{algo:algPoly} for resource sharing and allocation, no player or group of players has any incentive to cheat and misreport  capacities  and application requests, i.e., each player $n \in \mathcal{{N}}$ can achieve the highest payoff by truthfully reporting its capacity and requests.
\end{theorem}
\begin{proof}
Assume a player $n \in \mathcal{{N}}$ improves its  utility $\mathcal{U}_n$ by misreporting its capacity  and application requests, i.e.,   $\mathcal{U}'_n>\mathcal{U}_n$ where $\mathcal{U}'_n$ is the utility obtained by cheating and $\mathcal{U}_n$ is the utility obtained by truthful reporting. Each player receives its payoff for providing resources to its own applications first and then sharing the remaining resources (if any) with the applications of other service providers. Player payoff improves once  resources are provided and user requests are satisfied. 
    \begin{figure*}[h!]
	\centering
	\includegraphics[width=1.2\textwidth]{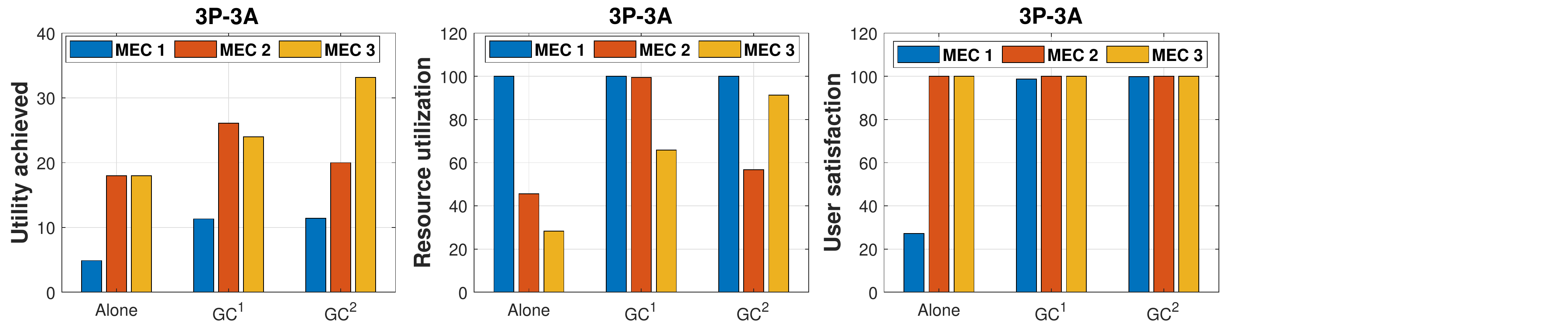}	
	\caption{Performance of our framework with Algorithm \ref{algo:alg2} for $3$ player - $3$ application settings.}
	\protect\label{fig:sat11}
	\vspace{-0.1in}
\end{figure*}

\begin{figure*}[h!]
	\centering
	\includegraphics[width=1.2\textwidth]{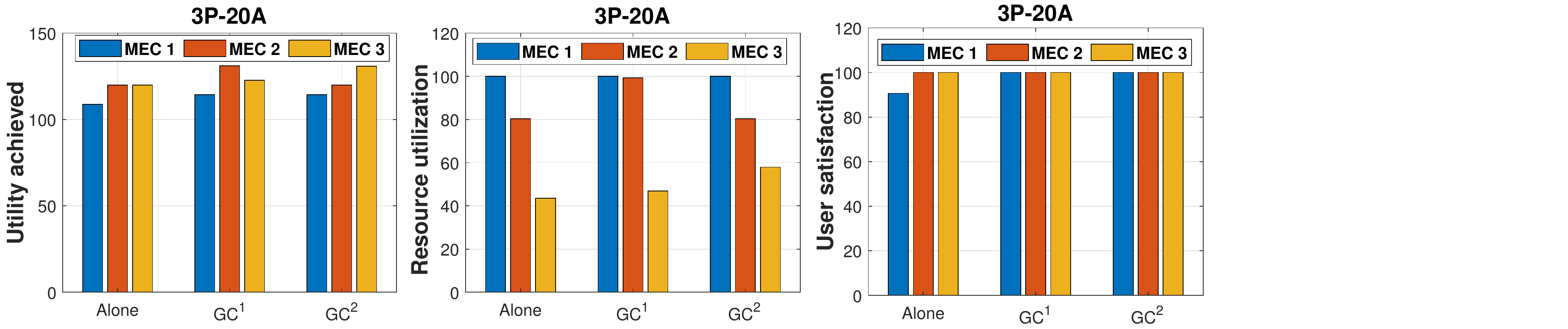}
	\caption{Performance of our framework with Algorithm \ref{algo:alg2} for $3$ player - $20$ application settings.}
	\protect\label{fig:sat12}
	\vspace{-0.1in}
\end{figure*}
\begin{algorithm}
	\begin{algorithmic}[]
		\State \textbf{Input}: $R, C,$ $\mathbf{u}, \boldsymbol{z}$ 
		\State \textbf{Output}: $\mathcal{J}, \mathcal{R},^{m}\hspace*{-0.05in}\boldsymbol{X}_n$
		\For{\texttt{$n \in \mathcal{G}_2$}}
		\For{\texttt{$m \in \mathcal{G}_1$}}
		\State \texttt{$\mathcal{J}_{m,n}\leftarrow$ Payoff for allocation \hspace*{0.4in} decision from \eqref{eq:opt_matching} }
		\State \texttt{$\mathcal{R}_{m,n}\leftarrow$ Resources allocated in \eqref{eq:opt_matching} }
		\EndFor
		\EndFor
	\end{algorithmic}
	\caption{Constructing matching matrix $\mathcal{J}$ }
	\label{algo:matrixConstruction}
\end{algorithm}
The larger the amount of resources provided and requests satisfied, the higher the obtained payoff will be. There are two cases for cheating:
\begin{itemize}
	\item \textbf{Under-reporting\footnote{Reporting a lower value than the actual value.} the capacity and application requests:} In this case, the player will not receive the  payoff possible by satisfying  its applications and fully utilizing its capacity by allocating resources to its own applications and applications of other players. This contradicts the assumption that $\mathcal{U}'_n>\mathcal{U}_n$. Hence, the player has no incentive to under-report its capacity and application  requests. 
	\item \textbf{Over-reporting\footnote{Reporting a larger value than the actual value.}  the capacity and application requests:}  Since the payoff depends on the actual amount of resources provided and requests satisfied rather than the reported capacity and application requests, the highest payoff possible for the player is $\mathcal{U}_n$ contradicting our assumption that  $\mathcal{U}'_n>\mathcal{U}_n$.     	\vspace{-0.1in}
\end{itemize}

\end{proof}

%% file: 04-exp-results.tex

We evaluate the performance of  proposed  resource sharing and allocation framework for  several parameter settings as shown in Table \ref{tab:settings}, where $\mathcal{{G}}_1$ and $\mathcal{{G}}_2$ represent the set of service providers with resource deficits and surpluses respectively. 
Each player has three different types of resources ($K=3$), i.e., storage, communication and computation. Without loss of generality, the model can be extended to include other types of resources. 
We use linear and sigmoidal utilities for all the players given below. 
\begin{align}\label{eq:utlinear}
u_n^j({x}_{n,k}^j)=a{x}_{n,k}^j+c
\end{align}
\begin{align}\label{eq:ut1}
u_n^j({x}_{n,k}^j)=\frac{1}{1+e^{-\mu(x_{n,k}^{j}-r^{j}_{n,k})}} 
\end{align}
$\mu$ is chosen to be $0.01$ whereas the constants $a$ and $c$ were randomly generated. 
Requests $R_{n}$ and capacities $C_{n}, \forall n \in \mathcal{N}$ are randomly generated for each setting within a pre-specified range. 
To show the advantage of resource sharing, we allocate larger capacities to certain players that   share  available resources with other domains and assist other players in meeting demand in order to increase their utilities.

\begin{table}[]
	\centering
	\caption{Simulation network settings.}
	\begin{tabular}{|l|l|l|l|}
		\hline
		\textbf{Setting}   & \bf{Parameters} & \textbf{$\mathcal{G}_1$}& \textbf{$\mathcal{G}_2$} \\ \hline
		\textbf{1} & $N=3, M_n=3,\forall n \in \mathcal{N}$   & \{1\}&  \{2, 3\}               \\ \hline
		\textbf{2} & $N=3, M_n=20,\forall n \in \mathcal{N}$     & \{1\}&    \{2, 3\}           \\ \hline
		\textbf{3} & $N=6, M_n=6,\forall n \in \mathcal{N}$         &\{1, 2, 3\}  &   \{4, 5, 6\}         \\ \hline
		\textbf{4} & $N=6, M_n=20,\forall n \in \mathcal{N}$           &\{1, 2, 5\} &     \{3, 4, 6\}   \\ \hline
	\end{tabular}
	\label{tab:settings}
\end{table}
  \begin{figure*}[h!]
	\centering
	\includegraphics[width=1.2\textwidth]{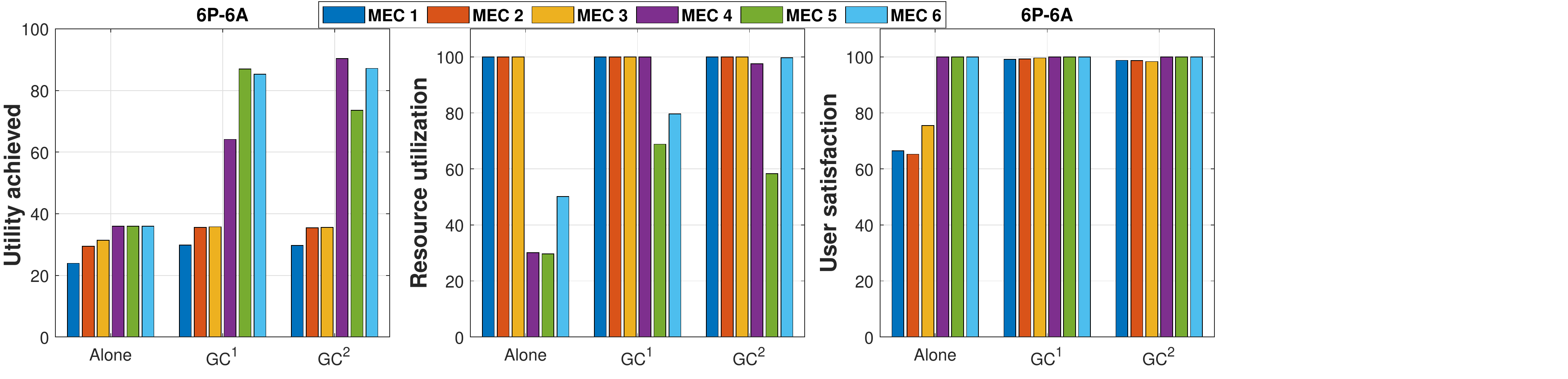}
	\caption{Performance of our framework with Algorithm \ref{algo:alg2} for  $6$ player - $6$ application settings.}
	\protect\label{fig:sat13}
	\vspace{-0.1in}
\end{figure*}

\begin{figure*}[h!]
	\centering
	\includegraphics[width=1.2\textwidth]{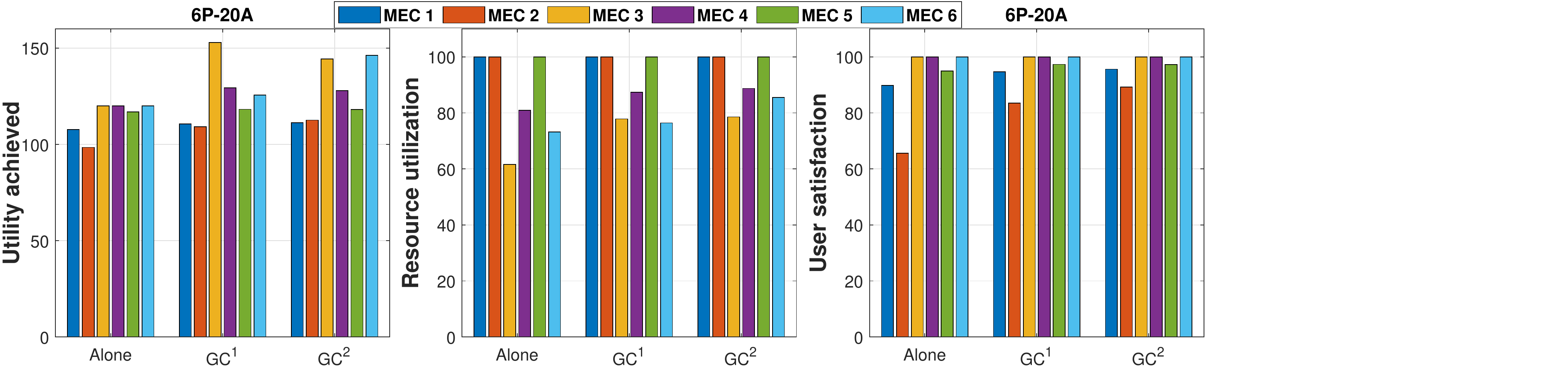}
	\caption{Performance of our framework with Algorithm \ref{algo:alg2} for  $6$ player - $20$ application settings.}
	\protect\label{fig:sat14}
	\vspace{-0.1in}
\end{figure*}
Simulations were run in \texttt{Matlab R2019b} on a \texttt{Core-i7} processor with \texttt{16 GB RAM}. To solve the optimization problems, we used the \texttt{OPTI-toolbox}\cite{currie2012opti}. Below, we provide detailed experimental results for both algorithms. 
\subsection{Results for Algorithm \ref{algo:alg2}}
\subsubsection{Verification of game-theoretic properties}
In Table \ref{tab:coalition}, we present results  for a 3-player 20-application game that verify different game theoretic properties  such as individually rationality, group rationality,  super additivity and show that the obtained allocation lies in the core. Player $1$ has a resource deficit while players $2$ and $3$ have  resource surpluses. 
 The payoffs all players receive in the grand coalition  are at least as good as they would receive working alone. This shows that the solution obtained using Algorithm \ref{algo:alg2} for the grand coalition is individually rational. Similarly, the value of coalition is the sum of payoff all players receive, hence our solution is group rational. Furthermore, as the coalition size increases, coalition value  increases. Hence, the grand coalition has the largest value, which shows the super additive nature of the game. Also,  no set of players has any incentive to divert from the grand coalition and form a smaller coalition. Hence, the grand coalition is stable and the allocation we obtain using Algorithm \ref{algo:alg2} lies in the core. There are two different results for the grand coalition, i.e., $GC^1$ or  $\{1, 2, 3\}$ and $GC^2$ or $\{1, 3, 2\}$ which shows that  changing the order of \texttt{for-loop}  in Step $5$ of Algorithm \ref{algo:alg2} only changes the utility achieved by players in $\mathcal{{G}}_2$. However, btoh  allocations   lie in the  core.  
 Similar results are seen for other settings given in Table \ref{tab:settings}, and hence are omitted here due to space constraints.  For $\{1, 2, 3\}$, player $2$ precedes player $3$ in Step $5$ (of Algorithm \ref{algo:alg2}) while player $3$ precedes player $2$ in Step $5$ for the grand coalition in $\{1, 3, 2\}$. From the results, it is evident that having  player $2$ execute Step $5$ in Algorithm \ref{algo:alg2} before player $3$ is preferable to the order player $3$ before player $2$. 
 However, it is impossible to know that in advance and we need to try all possible combinations.

\begin{table}[]
	\caption{Player payoff in different coalitions for a $3$ player - $20$ application game}
	\begin{tabular}{|l|l|l|l|l|}
		\hline
		\textbf{Coalition}       & \textbf{Player 1} & \textbf{Player 2} & \textbf{Player 3} & \textbf{Value of coalition}  \\ \hline
		\textbf{\{1\}}               & 108.74           & 0.00              & 0.00              & 108.74                 \\ \hline
		\textbf{\{2\}}               & 0             & 120            & 0              & 120                   \\ \hline
		\textbf{\{3\}}               & 0              & 0              & 120             & 120                   \\ \hline
		\textbf{\{1, 2\}}              & 111.51
		         & 131.01         & 0              & 242.52                \\ \hline
		\textbf{\{1, 3\}}              & 112.04           & 0              & 130.75            & 242.79               \\ \hline
		\textbf{\{2, 3\}}              & 0              & 120          & 120           & 240                 \\ \hline
		\textbf{\{1, 2, 3\}$^{1}$} & 114.39       & 131.01             &   122.75        & 368.15            \\ \hline
		\textbf{\{1, 3, 2\}$^{2}$} & 114.39       & 120            & 130.75           & 365.14            \\ \hline
	\end{tabular}
	\label{tab:coalition}
\end{table}

\subsubsection{Efficacy of the resource sharing framework}
Figures \ref{fig:sat11} and \ref{fig:sat12} show the efficacy of our resource sharing framework using Algorithm \ref{algo:alg2} and compare it with a setting in which the MECs are working alone for the 3 players and different application settings given in Table \ref{tab:settings}. We evaluate the impact of the proposed framework on MEC utility, request satisfaction, and resource utilization. It is evident that MECs with resource deficits are able to improve their application satisfaction and utilities (due to increase in application satisfaction) whereas MECs with resource surpluses improve their utilities by sharing their resources, which in turn increases  resource utilization.
For $GC^1$ results in the 6 player-6 application setting, player 4 and 5 precede player 6 in Step $5$ of Algorithm \ref{algo:alg2} whereas in $GC^2$, player 6 precedes players $4$ and $5$. For $GC^1$ results in the $6$ player-20 application setting, players 3 and 4 precede player 6 whereas in $GC^2$, player $6$ precedes players 3 and 4  in step 5 of Algorithm \ref{algo:alg2}. The utility function distributions among players vary for $GC^1$ and $GC^2$, however, both allocations are in the core and are Pareto optimal.  Figures \ref{fig:sat13} and \ref{fig:sat14} show the performance of Algorithm 1 in the 6 player, 6 and 20 application settings, respectively. It is evident from the results that resource sharing improves utility of MECs with resource surplus and user satisfaction of resource deficit MECs. 



\begin{figure*}
	\includegraphics[width=1.2\textwidth]{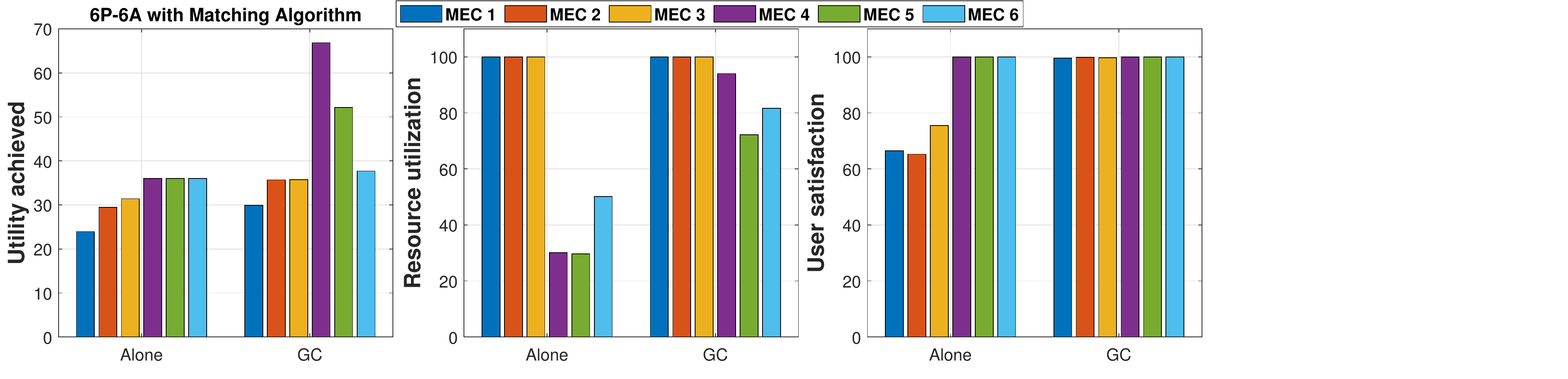}
	\caption{Performance of our framework with Algorithm \ref{algo:algPoly} for $6$ player-$6$ application settings.}
	\protect\label{fig:sat03}
	\vspace{-0.1in}
\end{figure*}
\begin{figure}
	\includegraphics[width=0.95\textwidth]{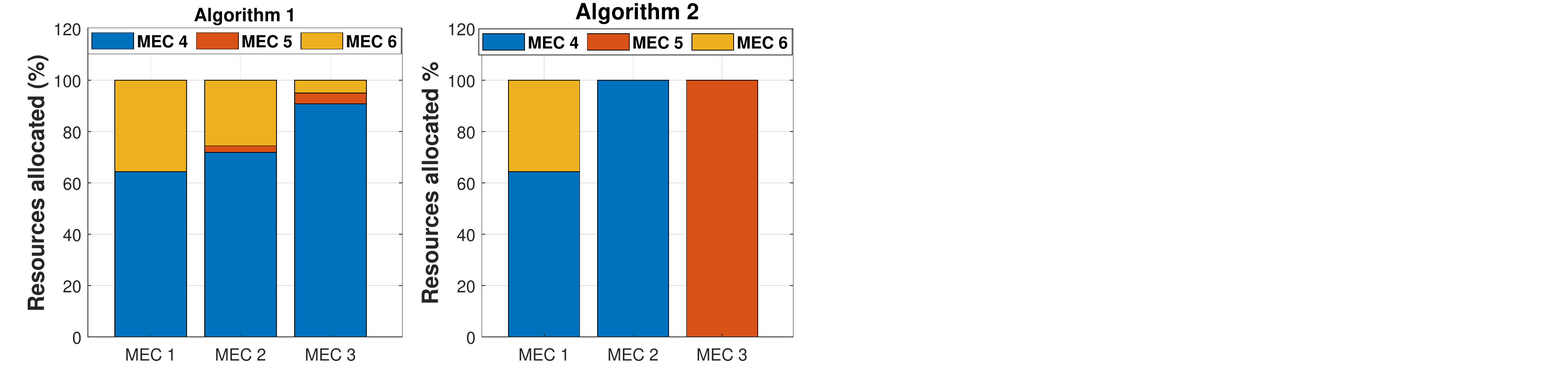}
	\caption{Resource fragmentation for Algorithms \ref{algo:alg2} and \ref{algo:algPoly} in 6 player-6 application setting.}
	\protect\label{fig:sat04}
	\vspace{-0.1in}
\end{figure}

\subsection{Results for Algorithm \ref{algo:algPoly}}
Figure \ref{fig:sat03} presents the results for Algorithm \ref{algo:algPoly}. It is evident that Algorithm \ref{algo:algPoly} improves the utilities and application satisfaction of different MECs (compared with MECs working alone). To highlight the decrease in resource fragmentation due to  Algorithm \ref{algo:algPoly}, we compare it with Algorithm \ref{algo:alg2}.  Figure \ref{fig:sat04} shows the percentage of resources allocated to applications of different MECs in $\mathcal{{G}}_1$ by MECs in $\mathcal{{G}}_2$ using Algorithms \ref{algo:alg2} and \ref{algo:algPoly}.  
 Using Algorithm \ref{algo:alg2} for resource sharing in the 6 player 6 application setting ($GC^1$),
applications of MEC 1 are allocated resources at MEC 4 and 6, whereas applications of both MEC 2 and MEC 3 receive resources from MECs 4, 5, 6. 
 For the same setting and using Algorithm \ref{algo:algPoly}, resources to applications of MEC 1 are provided by MEC 4  and 6. Resources to applications of MEC 2 and 3 are provided by MEC 4 and 5, respectively. 
This shows that resources fragmentation significantly reduces by Algorithm \ref{algo:algPoly} when compared with Algorithm \ref{algo:alg2} at the cost of added computational complexity.



The experimental results show that the proposed resource sharing framework can benefit MECs.  While Algorithm \ref{algo:alg2} is computationally more efficient,  it may result in comparatively larger resource fragmentation. On the other hand, Algorithm \ref{algo:algPoly}  reduces resource fragmentation  at the cost of computational complexity.

%% file: related.tex
There have been a number of solutions proposed in the literature related to resource allocation in edge computing \cite{you2017energy,wang2017computation,sardellitti2015joint,xu2017zenith,he2018s,zhang2017computing,xu2018game,liu2017incentive,zafari2018game}. Xu et al. \cite{xu2017zenith} propose a novel model for allocating resources in an edge computing setting where the allocation of distributed edge resources is decoupled from service provisioning management at the service provider side. The authors develop an auction-based resource sharing contract establishment and resource allocation that maximizes the utilities of service providers and edge computing infrastructure providers.  He et al. \cite{he2018s} relax the long-held assumption that storage resources are not shareable and study the optimal allocation of both shareable and non-shareable resources in a mobile edge computing setting. The authors consider the joint problem of service placement and request scheduling, and propose constant-factor approximation algorithm as the problem is proven to be NP-hard. Liu et al. \cite{liu2017incentive}  model the interaction among cloud service operators and edge service owners as a Stackelberg game for maximizing the utilities of both cloud and edge service providers. The authors obtain the optimal payment along with the computation offloading strategies. Zhang et al. \cite{zhang2017computing} propose an optimization framework, formulated as a Stackelberg equilibrium,  for edge nodes, data service operators and service subscribers that provides optimal resource allocation in a distributed manner. Xu et al. \cite{xu2018game} propose a secure caching scheme in heterogeneous networks for multi-homed subscribers that relies on a trust mechanism for verifying the reliability of edge computing-enabled small cell base stations. The authors also propose a Chinese reminder theorem  based protocol for preserving privacy. A Stackelberg game is used to model the interaction between the mobile users and the edge cloud and the goal is to maximize  utilities of both users and service providers. 
\par Our work in this paper differs from \cite{he2018s,xu2017zenith,xu2018game,liu2017incentive,zhang2017computing} as we consider  objectives of different service providers and allow resource sharing among  edge cloud service providers.  In \cite{zafari2019game}, we consider the strategy where the players do not differentiate between their own applications and applications of other players, and propose a centralized algorithm that provides an allocation from the core. 
This work is an extension of our earlier work in  \cite{zafari2018game}, where we model the resource sharing among mobile edge clouds as a canonical cooperative game with transferable utility (TU). However, in this work, we  model the problem as an NTU game which is more generic than the TU game \cite{roger1991game}. Unlike \cite{zafari2018game}, we do not rely on Shapley value, due to high complexity for obtaining a solution. Furthermore, we have not included application  satisfaction  and resource fragmentation as part of the resource sharing game in \cite{zafari2018game}.

%% file: 05-conclusions.tex
In this paper, we have proposed a cooperative game theory based resource sharing and allocation framework for edge clouds. We showed that for  monotonic, non-decreasing and non-negative positive utilities, resource sharing  among edge clouds can be modeled as a cardinally convex canonical  game. We have proved that the core exists and proposed  two efficient  algorithms that  provide allocation from the core. Hence, the obtained solutions are Pareto optimal and the grand coalition of the edge service providers is stable. Furthermore, we have reduced resource fragmentation using the matching based algorithm that also provides an allocation in the core. 
Experimental results show that  utilities of all service providers and   user satisfaction are improved,  when compared with  edge clouds working alone, using our framework.
